\documentclass[12pt,reqno]{amsart}
\usepackage{amsmath,amssymb,amsthm,epsfig,fullpage}
\date{\today}

\newcommand{\be}{\begin{eqnarray}}
\newcommand{\ee}{\end{eqnarray}}

\newcommand{\R}{\mathbb{R}}
\newcommand{\half}{\frac{1}{2}}
\newcommand{\ep}{\varepsilon}

\newcommand{\om}{\Omega}

\newcommand{\supp}{{\rm supp}\,}

\newcommand{\1}{{\bf 1}}

\newcommand{\diag}{{\rm diag}\,}
\newcommand{\n}{{\bf n}}
\newcommand{\N}{{\bf N}}
\newcommand{\F}{{\bf F}}
\newcommand{\x}{{\bf x}}
\newcommand{\U}{{\bf U}}
\newcommand{\V}{{\bf V}}
\newcommand{\z}{{\bf z}}

\newcommand{\m}{{\bf m}}

\newcommand{\B}{{\bf B}}

\newcommand{\Q}{{\bf Q}}

\newcommand{\az}{{\bf a}}

\newcommand{\e}{{\bf e}}

\newcommand{\y}{{\bf y}}
\newcommand{\A}{{\bf A}}

\newcommand{\Rz}{{\bf R}}

\newtheorem{thm}{Theorem}%[section]

\newtheorem{lem}[thm]{Lemma}

\graphicspath{%
    {converted_graphics/}% inserted by PCTeX
    {C:/Users/John/Documents/Papers/}% inserted by PCTeX
    {/}% inserted by PCTeX
    {C:/Users/John/Pictures/}% inserted by PCTeX
}
\makeindex             % used for the subject index
\title{Interaction of martensitic microstructures in adjacent grains}

\usepackage{graphicx}
\begin{document}

Submitted to Proceedings of ICOMAT 2017\vspace{.1in}

\maketitle
\vspace{-.2in}
\begin{center}%
   John M.  Ball\\
{\footnotesize Mathematical Institute, University of Oxford, Woodstock Road, 
Oxford OX2 6GG, U.K.}  
\\[5mm]\  Carsten Carstensen \\{\footnotesize Department of Mathematics, Humboldt-Universit\"{a}t zu
Berlin, Unter den Linden 6, D-10099 Berlin, Germany}
{\footnotesize }
\end{center}
\vspace{-.12in}
\date{empty}
\begin{abstract}
It is often observed that martensitic microstructures in adjacent polycrystal grains are related. For example, micrographs of Arlt \cite{Arlt90} (one reproduced in \cite[p225]{bhattacharya03a}) exhibit propagation of layered structures across grain boundaries in the cubic-to-tetragonal phase transformation in $\rm BaTiO_3$. Such observations are related to requirements of compatibility of the deformation at the grain boundary. Using a generalization of the Hadamard jump condition, this is explored in the nonlinear elasticity model of martensitic transformations for the case of a bicrystal with suitably oriented columnar geometry, in which the microstructure in both grains is assumed to involve just two martensitic variants, with a planar or non-planar interface between the grains. 
\end{abstract}
\vspace{.05in}

\noindent Keywords: Bicrystal, compatibility, grain boundary,  Hadamard jump condition.
\vspace{-.04in}

\section{Description of problem}
\numberwithin{equation}{section}
Consider a bicrystal consisting of two columnar grains  
$\om_1=\omega_1\times (0,d)$ (grain 1), $\om_2=\omega_2\times (0,d)$ (grain 2), where $d>0$ and $\omega_1, \omega_2\subset\R^2$ are bounded Lipschitz domains whose boundaries $\partial\omega_1, \partial\omega_2$  intersect nontrivially, so that   $\partial\omega_1\cap\partial\omega_2$ contains points in the interior   $\omega$  of $\overline\omega_1\cup\overline\omega_2$ (see Fig. \ref{bicrystal}). Let $\om=\omega\times(0,d)$. The interface between the grains is the set $\partial\Omega_1\cap\partial\om_2\cap\om=(\partial\omega_1\cap\partial\omega_2 \cap\omega)\times (0,d)$. Since by assumption the boundaries $\partial\omega_1, \partial\omega_2$ are locally the graphs of Lipschitz functions, and such functions are differentiable almost everywhere, the interface has  at almost every point (with respect to area) a well-defined normal $\n(\theta)=(\cos\theta,\sin\theta,0)$ in  the $(x_1,x_2)$ plane. We say that the interface is {\it planar} if it is contained in some plane $\{\x\cdot\n=k\}$ for a fixed normal $\n$ and constant $k$.
\begin{figure}[htp] % float placement: (h)ere, page (t)op, page (b)ottom, other (p)age
  \centering
  % file name: C:/Users/John/Documents/My PCTeX Files/TEXINPUT/bicrystal.jpg
  \includegraphics[width=4.42in,height=1.53in,keepaspectratio]{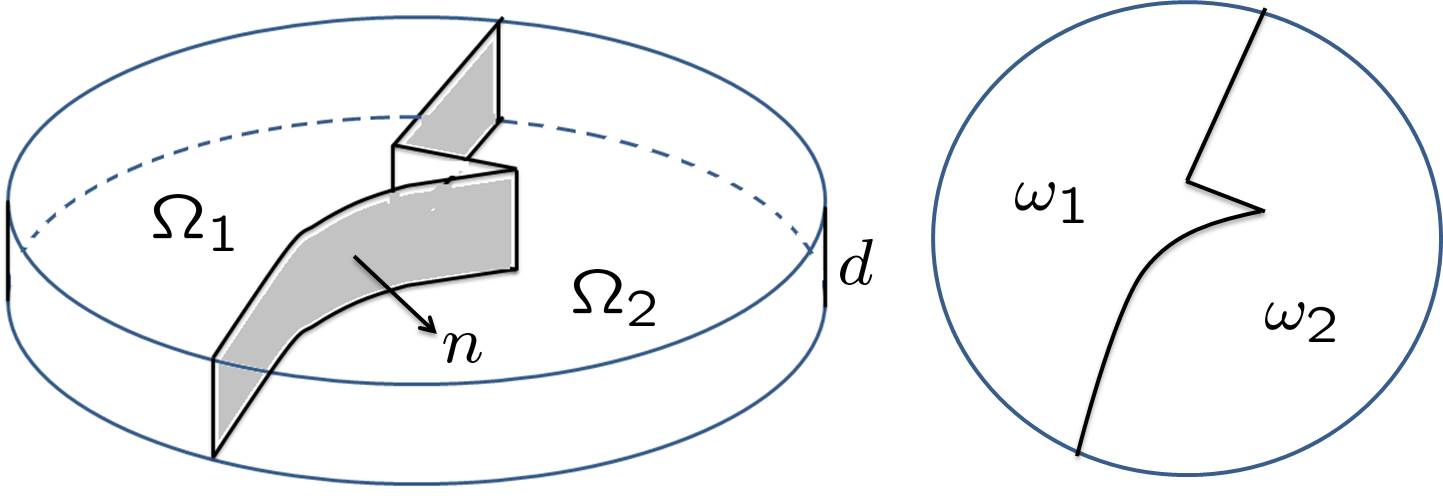}
  \caption{Bicrystal consisting of two grains $\Omega_1=\omega_1\times(0,d), \Omega_2=\omega_2\times(0,d).$}
  \label{bicrystal}
\end{figure}

We use the nonlinear elasticity model of martensitic transformations from \cite{j32,j40}, with corresponding free-energy density $\psi(\nabla \y, \theta)$ for  a single crystal at  temperature $\theta$ and deformation $\y=\y(\x)$ with respect to undistorted austenite at the critical temperature $\theta_c$ at which the austenite and martensite have the same free energy.   We denote by $\R^{n\times n}_+$ the set of real $n\times n$ matrices $\A$ with $\det\A>0$, and by $SO(n)$ the set of rotations in $\R^n$. At a fixed temperature $\theta<\theta_c$, we suppose that  
\be
\label{energywells}K=SO(3)\U_1\cup SO(3)\U_2 
\ee
is the set of  $\A\in\R^{3\times 3}_+$ minimizing $\psi(\A,\theta)$, where  $\U_1=\diag (\eta_2,\eta_1,\eta_3)$, $\U_2=\diag(\eta_1,\eta_2,\eta_3)$ and   $\eta_2>\eta_1>0$,   $\eta_3>0$. This corresponds to a tetragonal to  orthorhombic phase transformation (see \cite[Table 4.6]{bhattacharya03a}), or to an orthorhombic to monoclinic  transformation in which the transformation strain involves stretches of magnitudes $\eta_1, \eta_2$ with respect to perpendicular directions lying in the plane of two of the orthorhombic axes and making an angle of $\pi/4$ with respect to these axes\footnote{The general form of the transformation stretch for an orthorhombic to monoclinic transformation is given in \cite[Theorem 2.10(4)]{j40}. In general one can make a linear transformation of variables in the reference configuration which turns the corresponding energy wells into the form \eqref{energywells}. However, in \cite[Section 4.1]{j48} and the announcement of the results of the present paper in \cite{p35} it was incorrectly implied that the analysis based on $K$ as in \eqref{energywells} applies to a general orthorhombic to monoclinic transformation. This is not the case because the linear transformation in the reference configuration changes the deformation gradient corresponding to austenite in \cite{j48} and to the rotated grain in the present paper. A more general, but feasible, analysis would be needed to cover the case of general orthorhombic to monoclinic transformations.}.  Alternatively, for example, taking $\eta_3=\eta_1$ the analysis of this paper can be viewed as applying to a cubic to tetragonal transformation under the {\it a priori} assumption that only two variants are involved in the microstructure. 

We  suppose that $\om_1$   has cubic axes in the coordinate directions $\e_1, \e_2, \e_3$, while in $\om_2$   the cubic axes are rotated through an angle $\alpha$ about $\e_3$. By adding a constant to $\psi$ we may assume that $\psi(\A,\theta)=0$ for $\A\in K$. Then a zero-energy microstructure corresponds to a gradient Young measure\footnote{For an explanation of gradient Young measures and how they can be used to represent possibly infinitely fine microstructures see, for example, \cite{j56}.} $(\nu_\x)_{\x\in\om}$ such that 
\be 
\label{minmicro}
\supp \nu_\x\subset K \mbox{ for a.e. }\x\in\om_1,\;\;\supp\nu_\x\subset K\Rz_\alpha \mbox{ for a.e. }\x\in\om_2,
\ee
where $$\Rz_\alpha=\left(\begin{array}{lll} \cos\alpha&-\sin\alpha&0\\ \sin\alpha&\cos\alpha& 0\\
0&0&1\end{array}\right).$$     It is easily  shown that $K\Rz_\alpha=K$ if and only if $\alpha=n\pi/2$ for some integer $n$, and that $K\Rz_{\alpha+\pi/2}=K\Rz_\alpha$. We thus assume that $0<\alpha<\frac{\pi}{2}$, since this covers all nontrivial cases.

As remarked in \cite{p35}, by a result from \cite{bhattacharya92} there always exists a   zero-energy microstructure constructed using laminates, with gradient Young measure $\nu_\x=\nu$ satisfying \eqref{minmicro} that is independent of $x$ and has macroscopic deformation gradient $\bar\nu=\int_{\R^{3\times 3}_+}\A\,d\nu(\A)=(\det \U_1)\1$. Our aim is to give conditions on the deformation parameters $\eta_1,\eta_2,\eta_3$, the rotation angle $\alpha$ and the grain geometry which ensure that {\it any} zero-energy microstructure has a degree of complexity in each grain, in the sense that it does not correspond to a pure variant with constant deformation gradient in either of the grains.

\section{Rank-one connections between energy wells}
\label{rankone}
Let $\U=\U^T>0$, $\V=\V^T>0$. We say that the energy wells ${\rm SO}(3)\U$, ${\rm SO}(3)\V$ are {\it rank-one connected} if there exist $\Rz, \Q\in SO(3)$, $\az,\n\in\R^3$, $|\n|=1$ with $\Rz\U=\Q\V+\az\otimes\n$, where without loss of generality we can take $\Q=\1$. By the Hadamard jump condition this is equivalent to the existence of a continuous piecewise affine map $\y$ whose gradient $\nabla\y$ takes constant values $\A\in SO(3)\U$ and $\B\in SO(3)\V$ on either side of a plane with normal $\n$. The following is an apparently new version of a well-known result (see, for example, \cite[Prop. 4]{j32}, \cite[Theorem 2.1]{j56}, \cite{Khachaturyan83}),  giving necessary and sufficient conditions for two wells to be rank-one connected. A similar statement was obtained by  Mardare \cite{mardare}.
\begin{lem}\label{rankonea}
Let $\U=\U^T>0$, $\V=\V^T>0$. Then ${\rm SO}(3)\U$, ${\rm SO}(3)\V$ are rank-one connected if and only if
\be 
\label{conn}\U^2-\V^2= \gamma(\m\otimes \n+\n\otimes \m)
\ee
for unit vectors $\m$, $\n$ and some $\gamma\ne 0$. 
For suitable $\az_1,\az_2\in\R^3$ and 
 $\Rz_1,\Rz_2\in{\rm SO}(3)$, the rank-one connections between $\V$ and ${\rm SO}(3)\U$ are given by
\be 
\label{normals}
 \Rz_1\U=\V+\az_1\otimes \n,\quad \Rz_2\U=\V+\az_2\otimes \m.
\ee
\end{lem}
\noindent We omit the proof, which is not difficult. The main point of the lemma is that the normals corresponding to the rank-one connections are the vectors appearing in \eqref{conn}.  An interesting consequence is that if $\U,\V$ correspond to martensitic variants, so that $\V=\Q^T\U\Q$ for some $\Q\in SO(3)$, then, taking the trace in \eqref{conn} shows that the two possible normals are orthogonal (see \cite[Theorem 2.1]{j56}).

Using Lemma \ref{rankonea} we can calculate the rank-one connections between  $K$ and  $K\Rz_\alpha$. For example, for the rank-one connections between 
 ${\rm SO}(3)\,\U_1$ and ${\rm SO}(3)\,\U_1\Rz_\alpha$ we find that 
\be 
\label{1} 
\U_1^2-\Rz_\alpha^T\U_1^2\Rz_\alpha\,=(\eta_2^2-\eta_1^2)\sin\alpha\,(\m\otimes \n+\n\otimes \m)
\ee
where $\n=(\sin(\alpha/2),\cos (\alpha/2),0)$, $\m=(\cos (\alpha/2),-\sin (\alpha/2),0)$ so that the two possible normals $\n=(n_1,n_2,0)$ satisfy
$\tan\alpha= 2n_1n_2/(n_2^2-n_1^2)$.
Swapping $\eta_1$ and $\eta_2$ we see that the possible normals for rank-one connections between ${\rm SO}(3)\,\U_2$ and ${\rm SO}(3)\,\U_2\Rz_\alpha$ are the same. Similarly, we find that the possible normals for rank-one connections between 
 ${\rm SO}(3)\U_1$ and ${\rm SO}(3)\U_2\Rz_\alpha$ or between  ${\rm SO}(3)\U_2$ and ${\rm SO}(3)\U_1\Rz_\alpha$ satisfy
$\tan\alpha= (n_1^2-n_2^2)/2n_1n_2$.

\section{Main results}
\label{main}
Suppose  there exists a  gradient Young measure of the form \eqref{minmicro} such that $\nu_\x=\delta_\F$ for a.e. $\x\in \om_1$ for some $\F\in K$, corresponding to a pure variant in grain 1. It follows that the corresponding macroscopic gradient $\nabla\y(\x)=\bar\nu_\x=\int_{M^{3\times 3}_+}\A d\nu_x(\A)$ satisfies
\be 
\label{macro}
\nabla\y(\x)=\F \mbox{ for a.e. }\x\in\om_1, \;\;\nabla\y(\x)\in (K\Rz_\alpha)^{\rm qc} \mbox{ for a.e. }\x\in\om_2,
\ee 
where $E^{\rm qc}$ denotes the quasiconvexification of a compact set $E\subset \R^{3\times 3}$, that is the set of possible macroscopic deformation gradients corresponding to microstructures using gradients in $E$. As determined in \cite[Theorem 5.1]{j32} (and more conveniently in \cite[p155]{bhattacharya03a}) $K^{\rm qc}$ consists of those $\A\in \R_+^{3\times 3}$ with
\be 
\label{qc}
\A^T\A=\left(\begin{array}{ccc}a&c&0\\c&b&0\\ 0&0&\eta_3^2\end{array}\right)\mbox{ and }\;(a,b,c)\in P,
\ee 
\be 
\label{qc1}
P=\{(a,b,c): a>0,\; b>0,\; ab-c^2=\eta_1^2\eta_2^2,\; a+b+|2c|\leq \eta_1^2+\eta^2_2\},
\ee
and is equal to the polyconvexification $K^{pc}$ of $K$. It follows that $(K\Rz_\alpha)^{\rm qc}=K^{\rm qc}\Rz_\alpha=K^{\rm pc}\Rz_\alpha$. The proof of \eqref{qc} shows also that the quasiconvexification $\tilde K^{\rm qc}$ of 
$\tilde K=SO(2)\tilde\U_1\cup SO(2)\tilde\U_2$, 
where $\tilde\U_1=\left(\begin{array}{cc}\eta_2& 0\\ 0&\eta_1\end{array}\right)$, $\tilde\U_2=\left(\begin{array}{cc}\eta_1& 0\\ 0&\eta_2\end{array}\right)$,
is equal to $\tilde K^{\rm pc}$ and is given by the set of $\tilde\A\in \R^{2\times 2}_+$ such that 
$
\tilde\A^T\tilde \A=\left(\begin{array}{cc}a&c\\c&b\end{array}\right)\mbox{ for } (a,b,c)\in P.
$

Let $\x_0=(\tilde\x_0,\delta)$, where $\tilde\x_0\in\partial\omega_1\cap\partial\omega_2\cap\omega$, $0<\delta<d$, be such that the interface   has a well-defined normal $\n=(\cos\theta,\sin\theta,0)=(\tilde\n,0)$ at $\x_0$. By \cite{j39}, there exists $\ep>0$ such that in ${\mathcal U}:=B(\tilde\x_0,\ep)\times(\delta-\ep,\delta+\ep)$, $0<\ep<\delta$, the map $\y$ is a plane strain, that is 
\be 
\label{planestrain}
\y(\x)=\Rz(z_1(x_1,x_2),z_2(x_1,x_2),\eta_3 x_3 + \gamma)\mbox{ for a.e. }\x\in \mathcal U
\ee
 for some $\Rz\in SO(3)$, $\z:B(\tilde\x_0,\ep)\to\R^2$ and $\gamma\in\R$, where $B(\tilde\x_0,\ep)$ denotes the open ball in $\R^2$ with centre $\tilde\x_0$ and radius $\ep$. Without loss of generality we can take $\Rz=\1$. Then, for $\tilde\x=(x_1,x_2)\in B(\tilde\x_0,\ep)$  we have
\be 
\label{reduced}
\nabla \z(\tilde\x)=\tilde \F\in\tilde K \mbox{ for a.e. }\tilde\x\in \omega_1,\;\; \nabla \z(\tilde\x)\in (\tilde K\tilde\Rz_\alpha)^{\rm pc}=\tilde K^{\rm pc}\tilde\Rz_\alpha \mbox{ for a.e. }\tilde\x\in \omega_2,
\ee 
where  $\tilde\Rz_\alpha=\left(\begin{array}{cc}\cos\alpha&-\sin\alpha\\ \sin\alpha&\cos\alpha\end{array}\right)$. By a two-dimensional generalization of the Hadamard jump condition proved in \cite{u5} this implies that there exists $\tilde\A\in (\tilde K\tilde\Rz_\alpha)^{\rm pc}$ such that 
\be 
\label{r1}
\tilde \F-\tilde\A=\tilde\az\otimes\tilde\n 
\ee
for some $\tilde\az\in\R^2$. Conversely, if the interface is planar with normal  $\n=(\tilde\n,0)$, then the existence of $\tilde\A\in (\tilde K\tilde\Rz_\alpha)^{\rm pc}$ satisfying \eqref{r1} implies the existence of a gradient Young measure $\nu=(\nu_\x)_{\x\in\om}$ satisfying \eqref{macro}. Indeed there then exists a sequence of gradients $\nabla \z^{(j)}$ generating a gradient Young measure $(\mu_{\tilde\x})_{\tilde\x\in\omega}$ such that 
\be 
\label{r2}
\bar\mu_{\tilde\x}=\tilde\F \mbox{ for a.e. }\tilde \x\in\omega_1,\;\; \bar\mu_{\tilde\x}=\tilde\A \mbox{ for a.e. }\tilde\x\in\omega_2,
\ee
and then 
 $\nabla \y^{(j)}(\x)=\left(\begin{array}{ccc} z^{(j)}_{1,1}&z^{(j)}_{1,2}&0\\ z^{(j)}_{2,1}&z^{(j)}_{2,2}&0\\0&0&\eta_3\end{array}\right)$ 
generates such a gradient Young measure.

 It turns out that we can say exactly when it is possible to solve \eqref{r1}. Set $\tau:=\eta_2/\eta_1>1$, $s^*=(\tau^4-1)/(\tau^4+1)$ and define for $0\le s\le 1$ the $C^1$ convex increasing function
\be 
\label{f}
 f(s):= \left\{\begin{array}{cc}
(\tau^4+1-2\tau^2\sqrt{1-s^2})/(\tau^4-1) &\mbox{if }s\leq s^*,\\s&\mbox{if }s>s^*.\end{array}\right.
\ee
\begin{thm}
\label{char}
There exist $\tilde\F\in \tilde K$ and $\tilde \A\in \tilde K^{\rm pc}\Rz_\alpha$ with $\tilde\F-\tilde\A=\tilde\az\otimes \tilde\n$ for $\tilde\n=(\cos\theta,\sin\theta)$  and some $\tilde \az \in\R^2$ if and only if
\be 
\label{theta}
 |\cos 2\theta|\le f(|\cos 2(\alpha+\theta) |).
\ee
\end{thm}
\begin{proof}
It is easily checked that the existence of $\tilde\F\in SO(3)\tilde\U_i$ and $\tilde \A$ is equivalent to the existence of $(a,b,c)\in P$ such that $|\tilde \U_i\tilde\n^\perp|^2=aN_1^2+bN_2^2+2cN_1N_2$, where
$\tilde\n^\perp=(-n_2,n_1)$ and $\N=(N_1,N_2)=\tilde\Rz_\alpha\tilde\n^\perp$. That is 
\be 
\label{maxmin}
\mbox{ either }n_1^2\eta_1^2+n_2^2\eta_2^2 \mbox { or } n_2^2\eta_1^2+n_1^2\eta_2^2 \;\;\in [m_-(\N), m_+(\N)],
\ee
where $m_\pm(\N)=\mathop{\small \min}\limits^{\textstyle\small\max}_{(a,b,c)\in P}(aN_1^2+bN_2^2+2cN_1N_2)$. Changing variables to $x=a+b$ and $y=a-b$ we find that $m_+(\N)=\max_{(x,y)\in P_2}\psi_+(x,y)$,  $m_-(\N)=\min_{(x,y)\in P_2}\psi_-(x,y)$, where
\be 
\label{xy}
2\psi_\pm(x,y)=x+y(N_1^2-N_2^2)\pm2|N_1N_2|\sqrt{x^2-y^2-4\eta_1^2\eta_2^2},
\ee 
\[
P_2=\Big\{(x,y)\in\R^2:|y|\le\eta_2^2-\eta_1^2\mbox{ and } \sqrt{y^2+4\eta_1^2\eta_2^2}\le x\le\frac{y^2+4\eta_1^2\eta_2^2+(\eta_1^2+\eta_2^2)^2}{2(\eta_1^2+\eta_2^2)}\Big\}.\]
The region $P_2$ is bounded by the two arcs $C_1:=\{x=\sqrt{y^2+4\eta_1^2\eta_2^2} \}$ and $C_2:=\{x=(y^2+4\eta_1^2\eta_2^2+(\eta_1^2+\eta_2^2)^2)/(2(\eta_1^2+\eta_2^2))\}$, defined for $|y|\leq \eta^2_2-\eta_1^2$, which intersect at the points $(x,y)=(\eta_1^2+\eta_2^2, \pm(\eta_2^2-\eta_1^2))$. Note that $\psi_\pm(x,y)$ have no critical points in the interior of $P_2$. In fact it is immediate that $\nabla \psi_+$ cannot vanish, while $\nabla\psi_-(x,y)=0$  leads to $y=x(N_2^2-N_1^2)$ and hence to the contradiction $0=x^2-y^2-4x^2N_1^2N_2^2=4\eta_1^2\eta_2^2>0$. Thus the maximum and minimum of $\psi_\pm(x,y)$ are attained on either $C_1$ or $C_2$. After some calculations we obtain
\begin{eqnarray*}
m_+(\N)&=&\frac12(\eta_1^2+\eta_2^2+(\eta_2^2-\eta_1^2)|N_1^2-N_2^2|),\\
m_-(\N)&=&\left\{\begin{array}{ll}\frac{\eta_1^2\eta_2^2}{\eta_1^2+\eta_2^2}\Big(1+ \sqrt{1-|N_1^2-N_2^2|^2}\Big)&\mbox{ if }\:|N_1^2-N_2^2|\le s^*,\\
 \frac12(\eta_1^2+\eta_2^2-|N_1^2-N_2^2|(\eta_2^2-\eta_1^2))&\mbox{ if }\:|N_1^2-N_2^2|\ge s^*.
\end{array}\right.
\end{eqnarray*}
Noting that $m_-(\N)\leq \half(\eta_1^2+\eta_2^2)\leq m_+(\N)$,  that $m_+(\N)+m_-(\N)\leq \eta_1^2+\eta_2^2$, and that $\half(n_1^2\eta_1^2+n_2^2\eta_2^2)+\half(n_2^2\eta_1^2+n_1^2\eta_2^2)=\half(\eta_1^2+\eta_2^2)$ it follows that \eqref{maxmin} is equivalent to 
\be 
\label{min}
m_-(\N)\leq\min\{n_1^2\eta_1^2+n_2^2\eta_2^2, n_2^2\eta_1^2+n_1^2\eta_2^2\}=\half(\eta_1^2+\eta_2^2)-\half|n_2^2-n_1^2|(\eta_2^2-\eta_1^2).
\ee
With the relations $n_1^2-n_2^2=\cos 2\theta$ and $N_2^2-N_1^2=\cos 2(\theta+\alpha)$ this gives \eqref{theta}.
\end{proof}
\begin{thm}
\label{planar}
If the interface between the grains is planar then there always exists a zero-energy microstructure which is a pure variant in one of the grains.
\end{thm}
\begin{proof}
The case of a pure variant in grain 2 and a zero-energy microstructure in grain 1 corresponds to replacing $\theta$ by $\theta+\alpha$ and $\alpha$ by $-\alpha$ in the above. Hence, since $f(s)\geq s$,  if the conclusion of the theorem were false,  Theorem \ref{char} would imply that $|\cos 2\theta|> f(|\cos 2(\alpha+\theta)|)\geq|\cos 2(\theta+\alpha)|>f(|\cos 2\theta|)\geq|\cos 2\theta|$, a contradiction.
\end{proof}
\noindent Thus   to rule out having a pure variant in one  grain the interface {\it cannot} be planar. 
\begin{thm}
\label{twonormals}
There is no zero-energy microstructure which is a pure variant in one of the grains if the interface between the grains has a normal $\n^{(1)}=(\cos\theta_1, \sin\theta_1,0)\in E_1$ and a normal $\n^{(2)}=(\cos\theta_2, \sin\theta_2,0)\in E_2$, for the disjoint open sets
\be 
\nonumber E_1=\{\theta\in\R: f(|\cos 2(\theta+\alpha)|)<|\cos 2\theta|\},\;\;E_2=\{\theta\in\R:f(|\cos(2\theta)|) <|\cos 2(\theta+\alpha)|\}.
\ee
\end{thm}
\begin{proof}This follows immediately from Theorem \ref{char} and the preceding discussion.
\end{proof}
  In the case $\alpha=\pi/4$ the sets $E_1, E_2$ take a simple form (note that the normals not in $E_1\cup E_2$ correspond to the rank-one connections between the wells found in Section \ref{rankone}).
\begin{thm}
\label{piby4}
Let $\alpha=\pi/4$ and suppose\footnote {This extra condition, typically satisfied in practice, was accidentally omitted in the announcement in \cite{p35}.} that $\eta_2^2/\eta_1^2<1+\sqrt{2}$. Then 
\be
\label{D1a} E_1=\bigcup_j((4j-1)\pi/8,(4j+1)\pi/8),\;\;
 E_2=\bigcup_j((4j+1)\pi/8,(4j+3)\pi/8).
\ee
\end{thm}
\begin{proof}
Note that $\eta_2^2/\eta_1^2<1+\sqrt{2}$ if and only if $s^*< 1/\sqrt{2}$. Therefore if $|\sin 2\theta|<1/\sqrt{2}$ then $f(|\sin 2\theta|)<f(1/\sqrt{2})=1/\sqrt{2}<|\cos 2\theta|$. Hence $f(|\sin 2\theta|)<|\cos 2\theta|$ if and only if  $|\sin 2\theta|<1/\sqrt{2}$. This holds if and only if $\theta\in (\pi/2)\mathbb{Z}  + (-\pi/8, \pi/8)=E_1$. The case of $E_2$ is treated similarly.
\end{proof}\vspace{-.15in}

\section{Discussion} Compatibility across grain boundaries in polycrystals using a linearized elastic theory is discussed in \cite{bhattkohn97}, \cite[Chapter 13]{bhattacharya03a}. Whereas we use the nonlinear theory we are restricted to a very special assumed geometry and phase transformation. Nevertheless we are able to determine conditions allowing or excluding a pure variant in one of the grains without any {\it a priori} assumption on the microstructure (which could potentially, for example, have a fractal structure near the interface). This was possible using a generalized Hadamard jump condition from \cite{u5}. The restriction to a two-dimensional situation is due both to the current unavailability of a suitable three-dimensional generalization of such a jump condition, and because the quasiconvexification of the martensitic energy wells is only known for two wells. \vspace{.01in}

\noindent {\small {\bf Acknowledgement}.  The research of JMB was supported  by the EU (TMR contract FMRX - CT  EU98-0229 and
ERBSCI**CT000670), by 
EPSRC
(GRlJ03466, EP/E035027/1, and EP/J014494/1), the ERC under the EU's Seventh Framework Programme
(FP7/2007-2013) / ERC grant agreement no 291053 and
 by a Royal Society Wolfson Research Merit Award.}
 
\bibliography{gen2,balljourn,ballconfproc,ballprep}
\bibliographystyle{plain}

\end{document}